\documentclass[letterpaper, 10 pt, conference]{ieeeconf}  % Comment this line out
\IEEEoverridecommandlockouts                              % This command is only
\usepackage{graphics} % for pdf, bitmapped graphics files
\usepackage{epsfig} % for postscript graphics files
\usepackage{mathptmx} % assumes new font selection scheme installed
\usepackage{times} % assumes new font selection scheme installed
\usepackage{amsmath} % assumes amsmath package installed
\usepackage{amssymb}  % assumes amsmath package installed
\usepackage{color}

\usepackage[normalem]{ ulem }
\usepackage{soul}

\usepackage{todonotes}

\usepackage{cleveref}

\usepackage{graphicx}
\usepackage[font=scriptsize]{caption}
\usepackage[font=scriptsize]{subcaption}

\providecommand{\com[2]}{\begin{tt}[#1: #2]\end{tt}}

\newtheorem{prop}{Proposition}%[section]
\newtheorem{thm}{Theorem}%[section]
\newtheorem{defi}{Definition}%[section]
\newtheorem{cor}{Corollary}%[section]
\newtheorem{lem}{Lemma}%[section]
\newcommand{\prt}[1]{\left(#1\right)}

\newcommand{\R}{\mathbb{R}}
\newcommand{\E}{\mathbb{E}}

\newcommand{\Ep}[1]{\mathbb{E}\left[#1\right]}

\newcommand{\quotes}[1]{``#1''}

\title{\LARGE \bf
Lower bound performances for average consensus\\ in open multi-agent systems (extended version)
}

\author{Charles Monnoyer de Galland and Julien M. Hendrickx% <-this % stops a space
\thanks{ICTEAM institute, UCLouvain (Belgium). J.H. is also with the CISE, Boston University (USA). This work was supported by \quotes{Communaut\'e fran\c{c}aise de
Belgique - Actions de Recherche Concert\'ees}. C.M. is a FRIA fellow (F.R.S.-FNRS), and J.H. holds a WBI.World excellence fellowship. Email adresses: \texttt{charles.monnoyer@uclouvain.be}, \texttt{julien.hendrickx@uclouvain.be}.
}%
}

\begin{document}

\maketitle
\thispagestyle{empty}
\pagestyle{empty}

%%%%%%%%%%%%%%%%%%%%%%%%%%%%%%%%%%%%%%%%%%%%%%%%%%%%%%%%%%%%%%%%%%%%%%%%%%%%%%%%
%%%%%%%%%%%%%%%%%%%%%%%%%%%%%%%%%%%%%%%%%%%%%%%%%%%%%%%%%%%%%%%%%%%%%%%%%%%%%%%%
%%%%%%%%%%%%%%%%%%%%%%%%%%%%%%%%%%%%%%%%%%%%%%%%%%%%%%%%%%%%%%%%%%%%%%%%%%%%%%%%
\begin{abstract}
We derive fundamental limitations on the performances of intrinsic averaging algorithms in open multi-agent systems, which are systems subject to random arrivals and departures of agents. Each agent holds a value, and their goal is to estimate the average of the values of the agents presently in the system. We provide a lower bound on the expected Mean Square Error for any estimation algorithm, assuming that the number of agents remains constant and that communications are random and pairwise. Our derivation is based on the expected error obtained with an optimal algorithm under conditions more favorable than those the actual problem allows, and relies on an analysis of the constraints on the information spreading mechanisms in the system, and relaxations of these.
\end{abstract}

%%%%%%%%%%%%%%%%%%%%%%%%%%%%%%%%%%%%%%%%%%%%%%%%%%%%%%%%%%%%%%%%%%%%%%%%%%%%%%%%
%%%%%%%%%%%%%%%%%%%%%%%%%%%%%%%%%%%%%%%%%%%%%%%%%%%%%%%%%%%%%%%%%%%%%%%%%%%%%%%%
%%%%%%%%%%%%%%%%%%%%%%%%%%%%%%%%%%%%%%%%%%%%%%%%%%%%%%%%%%%%%%%%%%%%%%%%%%%%%%%%
\section{Introduction}
\label{Sec:Intro}

Multi-agent systems show great benefits for modeling and solving problems in various domains including sensor networks \cite{ApMAS:WSN,ApMAS:TrackingWSN}, vehicle coordination \cite{ApMAS:vehicle}, or social phenomena \cite{ApMAS:CTAvgOpinionDynamics,ApMAS:Opinion&Confidence}. Among the most cited properties of multi-agent systems are their flexibility, their scalability, or their robustness. Yet, most results around multi-agent systems stand for asymptotic properties under the convenient assumption that their composition remains unchanged. The increasing size of the systems challenges this assumption as it implies slower processes and higher probabilities of arrivals and departures, making those non-negligible. It is also challenged by the chaotic nature of some systems, where communications can be difficult, or happen at a time-scale highly comparable to that of the arrivals and departures: vehicles can for instance share a stretch of road before heading to different destinations in collaborative multi-vehicles systems. 

All those reasons bring up the emergence of the study of \emph{Open Multi-Agent Systems}. Results about closed systems do not easily extend to open ones: repeated arrivals and departures imply important differences in the design and analysis of such systems, and result in several challenges, see e.g. \cite{OMAS:Gossip,OMAS:MAX}. First, with the frequent arrivals and departures of agents, the size of an open system changes with time, making the analysis of its state challenging. Moreover, incessant perturbations impact the state, but also in some cases the objective pursued by the agents of the system: algorithms then have to adapt to that variable objective, making their design challenging, and the usual convergence cannot be achieved anymore.

%%%%%%%%%%%%%%%%%%%%%%%%%%%%%%%%%%%%%%%%%%%%%%%%%%%%%%%%%%%%%%%%%%%%%
%%%%%%%%%%%%%%%%%%%%%%%%%%%%%%%%%%%%%%%%%%%%%%%%%%%%%%%%%%%%%%%%%%%%%
\subsection{State of the art}

There is little theoretical analysis of open multi-agent systems, as most multi-agent results rely on the assumption that systems are closed. However, some results considered arrivals and departures, such as simulation-based analyses of social phenomena performed in \cite{OMAS:sociophysics}. Also, it has been shown in \cite{OMAS:Gossip} that systems subject to gossip interactions in open features can be analyzed through size-independent descriptors, whose evolution described by a dynamical system is shown to asymptotically converge to some steady states. Moreover, algorithm design has been explored for MAX-consensus problems with arrivals and departures in \cite{OMAS:MAX} through additional variables, and where performance was measured by the probability for the estimate to eventually converge if the system closes. Similarly, THOMAS architecture was designed to maintain connectivity into P2P networks \cite{OMAS:OpenP2P}. Openness was also considered in applications such as VTL for autonomous cars to deal with cross-sections \cite{ApOMAS:VTL}.

Nevertheless, up to now, efficiently studying performances of algorithms and analyzing open systems remains a challenge in general. As agents cannot instantaneously react to perturbations, obtaining an exact result is often out of range. Hence, the goal should be to remain near a time-varying objective, and usual convergence is no more relevant to be studied as it cannot be achieved. A step towards understanding open systems is the derivation of \textit{fundamental performance limitations}: lower bounds on the performances that are possible to achieve. Through those limitations, one can obtain some quality criterion for algorithm design in open systems, but also get a better understanding of the possible bottlenecks that could arise.

%%%%%%%%%%%%%%%%%%%%%%%%%%%%%%%%%%%%%%%%%%%%%%%%%%%%%%%%%%%%%%%%%%%%%
%%%%%%%%%%%%%%%%%%%%%%%%%%%%%%%%%%%%%%%%%%%%%%%%%%%%%%%%%%%%%%%%%%%%%
\subsection{Contribution}
We establish fundamental performance limitations in open multi-agent systems for intrinsic averaging consensus problems (\textit{i.e.} estimating the average of all intern values owned by agents \textit{present in the system at that time}), where information is exchanged between agents through random pairwise communications. In this analysis, we focus on systems of fixed size: we assume that each departure of an agent is instantaneously followed by an arrival, so that \emph{only replacements occur, and the system size remains constant}, see Section \ref{Sec:Statement} for a formal definition. The approach used for deriving the fundamental limitations is detailed in Section~\ref{Sec:Methodology}, and consists in evaluating the performances of an algorithm that is provably optimal under favorable settings where agents have access to more resources than decentralized problems typically allow. These performances depend on some complicated distribution related to information propagation. Bounds on this distribution were obtained by considering relaxed conditions for the information exchange, namely the Ping model in Section~\ref{Sec:Ping} and the Infection model in Section~\ref{Sec:Infection}, to properly define the performance limitations.

Since consensus problems are commonly a building block for various multi-agent complex applications (such as decentralized optimization \cite{DO:dualAvg} or several control applications \cite{ApMAS:vehicle}), we expect the techniques we use for deriving those limitations to be extendable to more advanced tasks on open multi-agent systems.

%%%%%%%%%%%%%%%%%%%%%%%%%%%%%%%%%%%%%%%%%%%%%%%%%%%%%%%%%%%%%%%%%%%%%%%%%%%%%%%%
%%%%%%%%%%%%%%%%%%%%%%%%%%%%%%%%%%%%%%%%%%%%%%%%%%%%%%%%%%%%%%%%%%%%%%%%%%%%%%%%
%%%%%%%%%%%%%%%%%%%%%%%%%%%%%%%%%%%%%%%%%%%%%%%%%%%%%%%%%%%%%%%%%%%%%%%%%%%%%%%%
\section{Problem statement}
\label{Sec:Statement}

%%%%%%%%%%%%%%%%%%%%%%%%%%%%%%%%%%%%%%%%%%%%%%%%%%%%%%%%%%%%%%%%%%%%%
%%%%%%%%%%%%%%%%%%%%%%%%%%%%%%%%%%%%%%%%%%%%%%%%%%%%%%%%%%%%%%%%%%%%%
\subsection{System description}

We consider a set of $N$ agents, labelled from $1$ to $N$. Each agent $i \in \{1\ldots N\}$ holds an i.i.d. intrinsic value $x_i(t) \in \R$ randomly selected from a distribution which we assume without loss of generality to be zero mean and of variance $\E x_i^2 = \sigma^2$. This value remains constant except at a replacement, which we model as the complete erasure of the agent's memory and the attribution of a new value $x_i(t)$ drawn from the same distribution. Those replacements happen according to a Poisson clock defined by an individual replacement rate $\lambda_r$, so that on average $N\lambda_r$ replacements happen in the whole system per unit of time.

Agents interact \textit{via} pairwise communications: every two agents interact with each other at random times determined by a Poisson clock with rate $\lambda_c$. There are thus on average $\frac{N(N-1)}{2}\lambda_c$ communications taking place in the system per unit of time. Agents are deterministic and have unlimited memory and computation power and there is no restriction on the size or nature of the messages they can exchange. They have a unique identifier that they can use and they have access to a common universal time. The way we model replacements as a memory erasure also means agents know the correspondence between the agents having left and those having replaced them since they share the same label. We assume moreover that they know the number of agents $N$, the parameters $\lambda_r$, $\lambda_c$, and the distribution of the $x_i$ (although we will see that the results would be the same if they knew only the expected value of $x_i$, assumed here to be $0$). Agents have thus access to significantly more information than in many works on multi-agent systems, but our lower bounds on the algorithm performance will of course also apply to these more usual situations, since our setting allows implementing any algorithm that could be implemented under more restrictive conditions.

%%%%%%%%%%%%%%%%%%%%%%%%%%%%%%%%%%%%%%%%%%%%%%%%%%%%%%%%%%%%%%%%%%%%%
%%%%%%%%%%%%%%%%%%%%%%%%%%%%%%%%%%%%%%%%%%%%%%%%%%%%%%%%%%%%%%%%%%%%%
\subsection{Objective}

In intrinsic averaging, agents try to estimate the average of the intern values of the agents present in the system at that time, denoted $\bar{x}(t)$. For that purpose, every agent $i$ maintains its own estimate $y_i(t)$ of that average based on its knowledge about the system, and updates it in continuous time. Under ideal conditions, one would expect that estimate to become $y_i(t) = \bar{x}(t) = \tfrac{1}{N}\sum_{j=1}^N x_j(t)$.

This is not achievable in open systems because of the variable objective and the delays to transmit information in the system, leading to the absence of usual convergence. 
We thus need a quantitative measure of the \textit{performance} of an algorithm solving this problem in open systems, and choose the classical Mean Square Error (MSE) of the estimation of all the agents at time $t$, presented in equation (\ref{eq:stt:avg_square_error}).
\begin{equation}
    \label{eq:stt:avg_square_error}
    C(t) := \tfrac{1}{N}\sum\nolimits_{j=1}^N (\bar{x}(t)-y_j(t))^2
\end{equation}

In this analysis, we focus on the steady state-error: we assume the system has been running since $-\infty$, so that the effect of initial conditions has disappeared, and $\Ep{C(t)}$ can be considered independent of the time (one can verify that our processes remain well defined). We thus derive fundamental performance limitations for all algorithms that can be implemented in our setting as a (time-invariant) lower bound on (\ref{eq:stt:crit_expected_value}). 
\begin{equation}
    \label{eq:stt:crit_expected_value}
    \Ep{C(t)} = \Ep{(\bar{x}(t)-y_i(t))^2}
\end{equation}

%%%%%%%%%%%%%%%%%%%%%%%%%%%%%%%%%%%%%%%%%%%%%%%%%%%%%%%%%%%%%%%%%%%%%
%%%%%%%%%%%%%%%%%%%%%%%%%%%%%%%%%%%%%%%%%%%%%%%%%%%%%%%%%%%%%%%%%%%%%
\subsection{Preliminary notions}
\label{sec:stt:defs}

Before stating our first result, we need to introduce certain notions related to the information available to an agent when computing its estimate, and to the age of its information.

We first formalize in the next definition all the information about events and values $x_j$ to which an agent $i$ could possibly access, and thus influence its estimate $y_i(t)$.

\begin{defi}
    \label{def:knowledge_set}
    At its arrival at time $t$, an agent $i$ whose value is set to $x_i$ owns a \textit{knowledge set} $\omega_i(t)$ such that $\omega_i(t^+)~=~\{i,x_i,t\}$. If agents $i$ and $j$ interact at time $\tilde t$, then $\omega_i(\tilde t^+)~=~\omega_j(\tilde t^+)~=~\omega_i(\tilde t^-)~\cup~\omega_j(\tilde t^-)~\cup~\{i~\hbox{-}~j;\tilde t, x_i(\tilde t),x_j(\tilde t)\}$.
\end{defi}
\vspace{0.1cm}

The last expression of the union defining an interaction denotes that $i,j$ have interacted at time $\tilde t$ and confirms their values at that time.\footnote{This confirmation is included for simplicity, but is actually redundant, as the values at that time could be obtained from $\omega_i(\tilde t^-)$ and $\omega_j(\tilde t^-)$.} 
In other words, the knowledge sets after interaction consist of the union of the knowledge sets, to which is added the information about the interaction and confirmation of the values. 
Standard results in distributed computation show that \textit{any estimate $y_i(t)$ that an agent can compute in our settings can actually be computed based only on $\omega_i(t)$ and the time $t$.}

Observe that $\omega_i(t)$ may contain various values $x_j(t')$ for different $t'<t$. However, we will see later that only the most recent known values about the other agents are needed to build the estimate $y_i(t)$.
\begin{defi}
    \label{def:most_recent_info}
    The \textit{most recent value} known by the agent $i$ about $j$ at time $t$ is denoted $\tilde x_j^{(i)}(t) = x_j\big(\tilde t_j^{(i)}(t)\big)$, where $\tilde t_j^{(i)}(t) := \max \{s : x_j(s)\in\omega_i(t)\}$ is 
    the time at which that most recent value was first obtained. We denote the \textit{age of that information} by $T_j^{(i)}(t) := t-\tilde t_j^{(i)}(t)$.
\end{defi}
\vspace{0.1cm}

By convention, if no value $x_j$ lies in $\omega_i(t)$, we set $x_j$ to $0$ and the corresponding age to $+\infty$. Moreover, the estimate of an agent about itself is always correct, and the corresponding age of the information is $0$.

In steady state, and due to the symmetry between the agents, the distribution of $T_j^{(i)}(t)$ is independent of $i$, $j$ and~$t$. Hence, all agents share a common \textit{cdf} and \textit{pdf} for that variable, respectively denoted $F(s)$ and $f(s)$. Moreover, with a small abuse of language, we say that another \textit{pdf} $f^*(s)$ \textit{bounds} $f(s)$ when the corresponding \textit{cdf} $F^*(s)$ satisfies $F^*(s)\geq F(s) \ \forall s$, which is always satisfied for a random variable $T^* \leq_{st} T$, in the \textit{usual stochastic order} \cite{MISC:UsualStochasticOrder}.

%%%%%%%%%%%%%%%%%%%%%%%%%%%%%%%%%%%%%%%%%%%%%%%%%%%%%%%%%%%%%%%%%%%%%%%%%%%%%%%%
%%%%%%%%%%%%%%%%%%%%%%%%%%%%%%%%%%%%%%%%%%%%%%%%%%%%%%%%%%%%%%%%%%%%%%%%%%%%%%%%
%%%%%%%%%%%%%%%%%%%%%%%%%%%%%%%%%%%%%%%%%%%%%%%%%%%%%%%%%%%%%%%%%%%%%%%%%%%%%%%%
\section{Bound in terms of the information spreading mechanism}
\label{Sec:Methodology}

We present in Theorem \ref{thm:general_bnd_pdf} a general lower bound on (\ref{eq:stt:crit_expected_value}) depending on the way information spreads in the system. To properly define a bound, we will thus need to find relaxations of that information spreading to instantiate the expression.

\vspace{0.1cm}
\begin{thm}
    \label{thm:general_bnd_pdf}
    %\comjh{Dire for any algorithm in the setting of }
    The performances of any algorithm in the setting defined in Section~\ref{Sec:Statement} are bounded from below as
    \begin{equation}
        \label{eq:thm:general_bnd_pdf}
        \Ep{C(t)} \geq \frac{N-1}{N^2}\int_0^\infty f^*(t) \prt{1-e^{-2\lambda_r t}} \sigma^2 \mathrm{d}t
    \end{equation}
    where $f^*(t)$ bounds $f(t)$ as defined in Section~\ref{sec:stt:defs}.
\end{thm}
\vspace{0.1cm}

The bound derived in the above theorem is obtained by studying an optimal algorithm in the setting of Section~\ref{Sec:Statement} (Section \ref{sec:cpt:optimal_algo}). It relies on the decomposition of the contributions of all the agents for the estimation of the average to obtain the MSE for some knowledge (\Cref{sec:cpt:decomposition,sec:cpt:single_agent_estimate,sec:cpt:indiv_error}), and then on the derivation of a steady state through the analysis of the way information spreads in the system (\Cref{sec:cpt:global_error,sec:cpt:pdf}).

%%%%%%%%%%%%%%%%%%%%%%%%%%%%%%%%%%%%%%%%%%%%%%%%%%%%%%%%%%%%%%%%%%%%%
\subsection{Optimal algorithm definition}
\label{sec:cpt:optimal_algo}

We first build an algorithm that is optimal for solving intrinsic averaging based on a knowledge set.

\begin{prop}
    \label{prop:best_algo}
    In the setting described in \Cref{Sec:Statement}, the following estimate is optimal in the sense of criterion (\ref{eq:stt:crit_expected_value}).
    %The best possible algorithm to estimate the average $\bar{x}(t)$ in the sense of the criterion (\ref{eq:stt:crit_expected_value}) computes for any knowledge set $\omega_i(t)$
    \begin{equation}
        \label{eq:prop:best_algo}
        y_i(t) = \Ep{\bar{x}(t)|\omega_i(t)} = \tfrac{1}{N}\sum\nolimits_{j=1}^N \Ep{x_j(t)|\omega_i(t)}
    \end{equation}
\end{prop}
\vspace{0.1cm}
\begin{proof}
    %Given some knowledge set $\omega_i(t)$,
    Conditional to the knowledge set $\omega_i(t)$, algorithm (\ref{eq:prop:best_algo}) is optimal since it minimizes the MSE:
    $$\tfrac{d}{dy_i(t)}\Ep{(\bar{x}(t)-y_i(t))^2|\omega_i(t)} = 2y_i(t) - 2\Ep{\bar{x}(t)|\omega_i(t)} = 0.$$
    
    Hence, since any other algorithm only depends on $\omega_i(t)$ and $t$, the result of any of them $y^*_i(t)$ is such that
    $$\Ep{(\bar{x}(t)-y_i^*(t))^2|\omega_i(t)} \geq \Ep{(\bar{x}(t)-y_i(t))^2|\omega_i(t)}.$$
    Considering the expected value on all $\omega_i(t)$, it follows:
    $$\Ep{\Ep{(\bar{x}(t)-y_i^*(t))^2|\omega_i(t)}} \geq \Ep{\Ep{(\bar{x}(t)-y_i(t))^2|\omega_i(t)}}.$$
\end{proof}
\vspace{0.1cm}

Since the algorithm (\ref{eq:prop:best_algo}) is optimal in the sense of criterion (\ref{eq:stt:crit_expected_value}), its performance 
\begin{equation}
    \label{eq:gen:best_algo_MSE}
    \Ep{C(t)} = 
    \Ep{\Ep{\prt{\bar{x}(t)-\Ep{\bar{x}(t)|\omega_i(t)}}^2|\omega_i(t)}}
\end{equation}
provides a lower bound on the performance of all algorithms that can be deployed in our settings.

%%%%%%%%%%%%%%%%%%%%%%%%%%%%%%%%%%%%%%%%%%%%%%%%%%%%%%%%%%%%%%%%%%%%%
\subsection{Individual contribution}
\label{sec:cpt:decomposition}

We show in the following proposition that expression (\ref{eq:gen:best_algo_MSE}) conveniently reduces to the analysis of a single agent.

\begin{prop}
    \label{prop:gen:decomposition}
    The criterion (\ref{eq:stt:crit_expected_value})
    for the algorithm (\ref{eq:prop:best_algo}) conditional to the knowledge set $\omega_i(t)$ reduces to
    \begin{equation}
        \label{eq:prop:decomposition}
        \small
        \Ep{\prt{\bar{x}(t)-y_i(t)}^2|\omega_i(t)} = \tfrac{1}{N^2}\sum\nolimits_{j=1}^N \Ep{(x_j(t)-\hat{x}_j(t))^2|\omega_i(t)}
    \end{equation}
    where $\hat{x}_j(t) := \Ep{x_j(t)|\omega_i(t)}$ (we lighten the notation by dropping the dependence of $\hat x_j^{(i)}(t)$ on $i$).
\end{prop}
\vspace{0.1cm}
\begin{proof}
    With the optimal algorithm (\ref{eq:prop:best_algo}), one has 
    $$y_i(t) = \tfrac{1}{N}\sum\nolimits_{j=1}^N\Ep{x_j(t)|\omega_i(t)}$$ 
    It follows that the error given $\omega_i(t)$ is written 
    {\small$$\Ep{(\bar{x}(t)-y_i(t))^2|\omega_i(t)} = \frac{1}{N^2} \Ep{\prt{\sum\nolimits_{j=1}^N \prt{x_j(t)-\hat{x}_j(t)}}^2\middle|\omega_i(t)}.$$}
    The absence of correlation between the agents values $x_i(t)$ finally allows to nullify the crossed-product terms of the squared sum, to obtain the final expression.
\end{proof}

%%%%%%%%%%%%%%%%%%%%%%%%%%%%%%%%%%%%%%%%%%%%%%%%%%%%%%%%%%%%%%%%%%%%%
\subsection{Single agent estimate}
\label{sec:cpt:single_agent_estimate}

We can then explicitly write an expression for the estimate of a single agent $\hat{x}_j(t)$.

\begin{prop}
    \label{prop:single_agent_estimate}
    There holds,
    \begin{equation}
        \label{eq:prop:single_agent_estimate}
        \hat{x}_j(t) = \Ep{x_j(t)|\omega_i(t)} = e^{-\lambda_rT_j^{(i)}}x_j(t-T_j^{(i)}).
    \end{equation}
    Hence, the estimate $\hat{x}_j(t)$ only depends on the most recent information about the agent $j$ (note that the time-dependence of $T_j^{(i)}(t)$ is removed to lighten the notations).
\end{prop}

\begin{proof}
    The most recent information we know about $j$ is given by Definition~\ref{def:most_recent_info}, and we have no information about whether it was replaced since then. Hence, denoting $R$ the event that $j$ has been replaced, and $\bar R$ that it has not,
    $$\hat{x}_j(t) = \Ep{x_j(t) | \bar R}\cdot P(\bar R) + \Ep{x_j(t)|R} \cdot \prt{1-P(\bar R)}.$$
    By definition, $\Ep{x_j(t)|R} = 0$. Then, from Poisson properties:
    $$\hat{x}_j(t) = 0 + \tilde x_j^{(i)}(t)\ e^{-\lambda_r\prt{t-\tilde t_j^{(i)}(t)}}$$
    and the conclusion follows from the definition of the age of the most recent information.
\end{proof}

The result above also stands when an agent estimates its own value or that of an agent for which it has no information, leading respectively to $\hat x_i^{(i)}(t) = x_i(t)$ and $\hat x_j^{(i)}(t)=0$.

%%%%%%%%%%%%%%%%%%%%%%%%%%%%%%%%%%%%%%%%%%%%%%%%%%%%%%%%%%%%%%%%%%%%%
\subsection{Individual error}
\label{sec:cpt:indiv_error}

For concision matters, we denote
\begin{equation}
    \label{eq:gen:individual_MSE_def}
    C_j^{(i)}(t) := (x_j(t)-\hat{x}_j(t))^2
    %\Ep{C_j} := \Ep{(x_j(t)-\hat{x}_j(t))^2|\omega_i(t)}
\end{equation}
We can then write the MSE of the estimation of a single agent by injecting (\ref{eq:prop:single_agent_estimate}) into (\ref{eq:gen:individual_MSE_def}).

\begin{prop}
    \label{prop:individual_MSE}
    The MSE when estimating a single agent for algorithm (\ref{eq:prop:best_algo}) conditional to the knowledge set $\omega_i(t)$ is
    %Using the optimal algorithm (\ref{eq:prop:best_algo}) for some knowledge set $\omega_i(t)$, the MSE of estimation of a single agent $j$ is characterized by the age of the information about it $T_j$:
    \begin{equation}
        \label{eq:prop:individual_MSE}
        \Ep{\prt{x_j(t)-\hat{x}_j(t)}^2|\omega_i(t)} 
        %= \prt{1-e^{-2\lambda_rT_j^{(i)}}}\sigma^2
        = \Big(1-e^{-2\lambda_rT_j^{(i)}}\Big)\sigma^2
    \end{equation}
    and is thus entirely characterized by the age of the most recent information about it.
\end{prop}
\vspace{0.1cm}
\begin{proof}
    Denoting $R$ the event of at least one replacement of the agent $j$ during $T_j^{(i)}$ and by $\bar R$ the event of no replacement, we can develop (\ref{eq:gen:individual_MSE_def}) with a case-by-case analysis:
    $$\Ep{C_j^{(i)}(t)|\omega_i(t)} = \Ep{C_j^{(i)}(t)|R}\cdot P(R) + \Ep{C_j^{(i)}(t)|\bar R}\cdot P(\bar R).$$
    We develop each term to obtain the final result:
    \begin{align*}
        &P(R) = 1-e^{-\lambda_rT_j^{(i)}} &&\Ep{C_j^{(i)}(t)|R} = \Big(1+e^{-2\lambda_r T_j^{(i)}}\Big)\sigma^2\\
        &P(\bar R) = e^{-\lambda_rT_j^{(i)}} &&\Ep{C_j^{(i)}(t)|\bar R} = \Big(1-e^{-\lambda_rT_j^{(i)}}\Big)^2\sigma^2
    \end{align*}
\end{proof}

This last result allows to write the following:
\begin{equation}
    \label{eq:gen:MSE_given_omega}
    \Ep{C(t)|\omega_i(t)} = \tfrac{1}{N^2}\sum\nolimits_{j=1}^N \Big(1-e^{-2\lambda_rT_j^{(i)}}\Big)\sigma^2.
\end{equation}

%%%%%%%%%%%%%%%%%%%%%%%%%%%%%%%%%%%%%%%%%%%%%%%%%%%%%%%%%%%%%%%%%%%%%
\subsection{Global expected value}
\label{sec:cpt:global_error}

We have developed an expression for the error in terms of the age of the most recent information about the other agents in (\ref{eq:gen:MSE_given_omega}). We can now obtain an expression for criterion (\ref{eq:stt:crit_expected_value}) by computing the expected value of that result.

\begin{prop}
    \label{prop:steady_state}
    The MSE defined in (\ref{eq:stt:crit_expected_value}) is given by
    \begin{equation}
        \label{eq:prop:steady_state}
        \Ep{\bar{x}(t)-y_i(t)}^2 = \frac{N-1}{N^2}\int_0^\infty f(t) \prt{1-e^{-2\lambda_rt}}\sigma^2\ \mathrm{d}t,
    \end{equation} 
    where we remind that $\E x_i = 0$ and $\E x_i^2 = \sigma^2$.
    It is thus entirely characterized by the distribution of the age of an information, and knowing $f(t)$ leads to a proper bound.
\end{prop}
\vspace{0.1cm}
\begin{proof}
    By definition, the global MSE is 
    $$\Ep{(\bar{x}(t)-y_i(t))^2} = \Ep{\Ep{(\bar{x}(t)-y_i(t))^2|\omega_i(t)}}.$$
    Using the result (\ref{eq:gen:MSE_given_omega}), it becomes
    $$\Ep{(\bar{x}(t)-y_i(t))^2} = \tfrac{1}{N^2}\sum\nolimits_{j=1}^N \Ep{\prt{1-e^{-2\lambda_rT_j^{(i)}}}}\sigma^2.$$
    Finally, since all $T_j^{(i)}$ follow $f(t)$ from Section \ref{sec:stt:defs},
    $$\Ep{\prt{1-e^{-2\lambda_rT_j^{(i)}}}} = \int_0^\infty f(t) \prt{1-e^{-2\lambda_r t}} \mathrm{d}t,$$
    and the conclusion is then direct, where we remind that $T_i^{(i)} = 0$ by definition (\textit{cfr} Section~\ref{sec:stt:defs}).
\end{proof}

%%%%%%%%%%%%%%%%%%%%%%%%%%%%%%%%%%%%%%%%%%%%%%%%%%%%%%%%%%%%%%%%%%%%%
\subsection{Relaxation of the communication process}
\label{sec:cpt:pdf}

The previous result is actually that of Theorem~\ref{thm:general_bnd_pdf} where the \textit{pdf} is exactly $f(s)$. Since this \textit{pdf} reveals to be hard to compute, we show in the next proposition that a \textit{pdf} $f^*(s)$ bounding $f(s)$ as defined in Section~\ref{sec:stt:defs} leads to a proper lower bound on (\ref{eq:prop:steady_state}), which concludes the development of Theorem~\ref{thm:general_bnd_pdf}. In the next two Sections, we will then present two possible relaxations of the information spreading mechanism that will both lead to an appropriate $f^*(s)$ to instantiate the bound.

\begin{prop}
    \label{prop:CDF}
    Given some value $E = \int_0^\infty F'(t)\cdot err(t) \mathrm{d}t$ where $F(t)$ is a \textit{cdf}, and where $err(t)$ is a positive non-decreasing function, then for any other \textit{cdf} $F^*(t) \geq F(t) \ \forall t$,
    \begin{equation}
        \label{eq:General:Integral_higher_CDF_prop}
        E^* = \int_0^\infty (F^*)'(t)\cdot err(t) \mathrm{d}t \leq E
    \end{equation}
\end{prop}
\begin{proof}
    By defining $F^* = F + \Delta$ with $\Delta \geq 0$, one has $$E^* = E - \int_0^\infty \Delta\cdot err'(t)\mathrm{d}t$$ thanks to \textit{cdf} properties. Using then the positivity of $\Delta$ and $err'(t)$, the conclusion is direct.
\end{proof}

%%%%%%%%%%%%%%%%%%%%%%%%%%%%%%%%%%%%%%%%%%%%%%%%%%%%%%%%%%%%%%%%%%%%%%%%%%%%%%%%
%%%%%%%%%%%%%%%%%%%%%%%%%%%%%%%%%%%%%%%%%%%%%%%%%%%%%%%%%%%%%%%%%%%%%%%%%%%%%%%%
%%%%%%%%%%%%%%%%%%%%%%%%%%%%%%%%%%%%%%%%%%%%%%%%%%%%%%%%%%%%%%%%%%%%%%%%%%%%%%%%
\section{Strong assumption: Ping model}
\label{Sec:Ping}

%%%%%%%%%%%%%%%%%%%%%%%%%%%%%%%%%%%%%%%%%%%%%%%%%%%%%%%%%%%%%%%%%%%%%
%%%%%%%%%%%%%%%%%%%%%%%%%%%%%%%%%%%%%%%%%%%%%%%%%%%%%%%%%%%%%%%%%%%%%
\subsection{Assumption description}

The first relaxation we consider relies on a strong simplification of the communication mechanism. 
We assume that each time it communicates, an agent \textit{acquires all the information about all the agents presently in the system,} and never forgets it, even at replacements. We refer to this assumption as the \textit{Ping model}, in reference to the \textit{ping} software used to test the reachability of machines in a network. The most recent value known by the agent $i$ about $j$ at time $t$ is then given by $\tilde x_j^{(i)}(t) = x_j\big(\tilde t_j^{(i)}(t)\big)$ where $\tilde t_j^{(i)}(t)$ is the time of the very last communication in which agent $i$ was involved, and the age of that information $T_{j,Ping}^{(i)}(t)~=~t-~\tilde t_j^{(i)}(t)$ is the time spent since then.

\begin{prop}
    \label{prop:Ping:assumption}
    In steady state, $\forall i\neq j$, the age of the most recent information about $j$ held by $i$, noted $T_{j,Ping}^{(i)}(t)$ follows
        \begin{equation}
        f^{Ping}(s) = (N-1)\lambda_ce^{-(N-1)\lambda_cs},
    \end{equation}
    which bounds $f(s)$ as defined in Section~\ref{sec:stt:defs}.
\end{prop}

\begin{proof}
    The random variable $T_{j,Ping}^{(i)}(t)$ defines the time spent since the last communication involving the agent $i$, and thus follows by definition of the communications a Poisson process of rate $(N-1)\lambda_c$.
    Its \textit{cdf} is then given by
    $$F^{Ping}(s) = P\Big[T_{j,Ping}^{(i)}(t)\leq s\Big] = 1-e^{-(N-1)\lambda_cs}.$$
    The \textit{pdf} is then obtained with $f^{Ping}(s) = \tfrac{dF^{Ping}(s)}{ds}.$
    
    $T_{j,Ping}^{(i)}(t)$ is actually the minimal value that can take $T_j^{(i)}(t)$. Hence $T_{j,Ping}^{(i)}(t) \leq_{st} T_j^{(i)}(t)$, and $f^{Ping}(s)$ bounds $f(s)$.
\end{proof}

%%%%%%%%%%%%%%%%%%%%%%%%%%%%%%%%%%%%%%%%%%%%%%%%%%%%%%%%%%%%%%%%%%%%%
%%%%%%%%%%%%%%%%%%%%%%%%%%%%%%%%%%%%%%%%%%%%%%%%%%%%%%%%%%%%%%%%%%%%%
\subsection{Results}

\begin{thm}
    \label{thm:Ping:Bound}
    The performances of any algorithm in the setting defined in Section~\ref{Sec:Statement} are bounded from below as
    \begin{equation}
        \label{eq:Ping:bound}
        \Ep{C(t)} \geq \frac{N-1}{N^2}\prt{\frac{1}{1+\frac{N-1}{2}\frac{\lambda_c}{\lambda_r}}}\sigma^2
    \end{equation}
\end{thm}
\vspace{0.1cm}
\begin{proof}
    Applying Theorem~\ref{thm:general_bnd_pdf} with the \textit{pdf} $f^{Ping}(t)$, the bound is given by performing the integration
    $$\Ep{C(t)} \geq \tfrac{N-1}{N^2}\int_0^\infty (N-1)\lambda_ce^{-(N-1)\lambda_cs}\prt{1-e^{-2\lambda_rs}}\sigma^2 \mathrm{d}s.$$
    A few algebraic operations lead to the conclusion.
\end{proof}
\vspace{0.1cm}
\begin{figure}[thpb]
    \centering
    \includegraphics[width=0.45\textwidth]{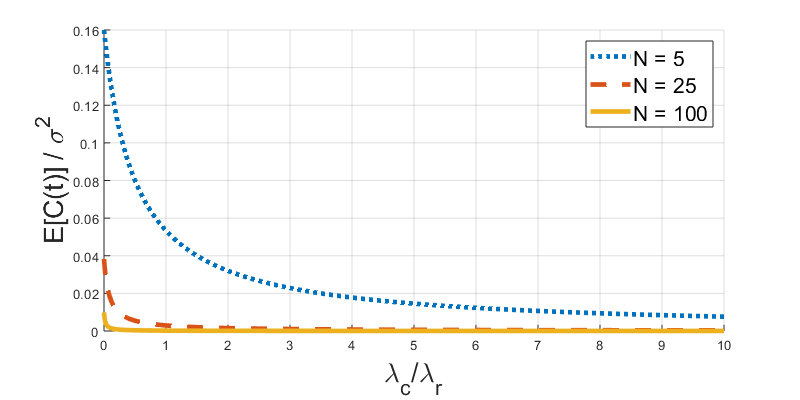}
    \caption{Bound derived with the Ping model (\ref{eq:Ping:bound}) in terms of the rate ratio between pairwise communications and replacements. Each curve stands for a different number of agents $N$, with a unit variance $\sigma^2=1$.}
    \label{fig:Ping_sevN_std}
\end{figure}

Several observations on the evolution of the error in terms of the ratio of the rates $\lambda_c/\lambda_r$ and $N$ arise from Fig.~\ref{fig:Ping_sevN_std}. When communications are rather rare ($\lambda_c/\lambda_r \rightarrow 0$), the error tends to $\tfrac{N-1}{N^2}\sigma^2$, \textit{i.e.} the error that an agent would obtain only considering itself in the estimation. As the communications become more frequent, the second factor of (\ref{eq:Ping:bound}) makes the bound decrease, to ultimately bring it to $0$ as $\lambda_c/\lambda_r\rightarrow+\infty$, which is the error to which we converge in a close system. 

Observe that $\bar{\lambda}_c = (N-1)\lambda_c$ is the communication rate for one agent. Hence, the bound is actually separated into two factors: $\tfrac{N-1}{N^2}\sigma^2$ and the second one that only depends on $\bar{\lambda}_c/\lambda_r$ which characterizes the expected number of interactions involving one agent before it is replaced.

%%%%%%%%%%%%%%%%%%%%%%%%%%%%%%%%%%%%%%%%%%%%%%%%%%%%%%%%%%%%%%%%%%%%%%%%%%%%%%%%
%%%%%%%%%%%%%%%%%%%%%%%%%%%%%%%%%%%%%%%%%%%%%%%%%%%%%%%%%%%%%%%%%%%%%%%%%%%%%%%%
%%%%%%%%%%%%%%%%%%%%%%%%%%%%%%%%%%%%%%%%%%%%%%%%%%%%%%%%%%%%%%%%%%%%%%%%%%%%%%%%
\section{Weaker assumption: Infection model}
\label{Sec:Infection}

%%%%%%%%%%%%%%%%%%%%%%%%%%%%%%%%%%%%%%%%%%%%%%%%%%%%%%%%%%%%%%%%%%%%%
%%%%%%%%%%%%%%%%%%%%%%%%%%%%%%%%%%%%%%%%%%%%%%%%%%%%%%%%%%%%%%%%%%%%%
\subsection{Assumption description}

The random variable $T_j^{(i)}(t)$ denotes the age of the most recent information about $j$ available to $i$ at time $t$, \textit{i.e.} the time since this information was emitted by $j$.  
Since the distribution of $T_j^{(i)}(t)$ is independent of the time $t$, its \textit{cdf} $F(s)$ corresponds to the probability that an information $x_j(\tau)$ (with $0 \leq \tau \leq s$) is available to $i$ at time $s$. 
This probability is lower than what it would be if the agents were never replaced (as information is erased with replacements). Hence we use this latter replacement-free situation as a relaxation to obtain a bound on the \textit{cdf}.
One can verify that when agents are never replaced, the information propagation follows a simple infection process: $j$ is infected at time $0$, and agents get infected as soon as they communicate with an infected agent.
The probability of $i$ having that information about $j$ at time $s$ is then equal to the probability of $i$ being infected at time $s$,
denoted by $P[T_{j,Inf}^{(i)}(t)\leq s]$ where $T_{j,Inf}^{(i)}$ is the time of its first infection.
Based on this, one can prove that the \textit{cdf} of $T_{j,Inf}^{(i)}$ bounds $F(s)$ as defined in Section~\ref{sec:stt:defs}.

\begin{prop}
    \label{prop:inf_valid}
    In steady state, the pdf of $T_{j,Inf}^{(i)}(t)$, noted $f^{Inf}(s)$, bounds $f(s)$ as defined in Section~\ref{sec:stt:defs}.
\end{prop}
\begin{proof}
    Considering one realization of communications, there exists a sequence of events where $T_j^{(i)}(t)$ is such that it did not suffer replacements. Since that sequence always also exists with the Infection model, $T_{j,Inf}^{(i)}(t) \leq_{st} T_j^{(i)}(t)$.
\end{proof}

%%%%%%%%%%%%%%%%%%%%%%%%%%%%%%%%%%%%%%%%%%%%%%%%%%%%%%%%%%%%%%%%%%%%%
%%%%%%%%%%%%%%%%%%%%%%%%%%%%%%%%%%%%%%%%%%%%%%%%%%%%%%%%%%%%%%%%%%%%%
\subsection{Result}

\begin{prop}
    \label{prop:Infection_assumption}
    In steady state, $\forall i\neq j$, the age of the most recent information agent $i$ has about agent $j$ with the Infection model, denoted $T_{j,Inf}^{(i)}(t)$, follows the \textit{cdf}
    \begin{equation}
        \label{eq:Infection:CDF}
        F^{Inf}(s) = \sum\nolimits_{k=1}^N \tfrac{k-1}{N-1}P_k(s)
    \end{equation}
    where $P_k(s)$ is defined by an ODE system:
    \begin{equation}
        \label{eq:Infection:Pk}
        \dot{P}_k(s) = (k-1)(N-k+1)\lambda_cP_{k-1}(s) - k(N-k)\lambda_cP_k(s)
    \end{equation}
    such that $P_1(0)=1$ and $P_{k\neq1}(0)=0$.
\end{prop}
\begin{proof}
    Let $I(s)$ be the set containing the labels of the infected agents at time $s$, $N_{I(s)}$ its size, and denote $P_k(s)~=~P[N_{I(s)}~=~k]$. Then, one has
    $$F^{Inf}(s) = P[i\in I(s)] = \sum\nolimits_{k=1}^N P_k(s)\cdot P[i\in I(s)|N_{I(s)}=k].$$
    
    From the symmetry between the agents, and since they are all interchangeable, one has 
    $$P[i\in I(s) | N_{I(s)}=k] = \tfrac{k-1}{N-1}.$$
    
    Finally, the ODE system defining $P_k(s)$ is a standard result from continuous time Markov chains theory.
\end{proof}
\vspace{0.1cm}

\begin{thm}
    \label{thm:Infection:Bound}
    The performances of any algorithm in the setting defined in Section~\ref{Sec:Statement} are bounded from below as
    \begin{equation}
        \label{eq:Infection:Bound}
        \Ep{C(t)} \geq \tfrac{N-1}{N^2}\prt{1-w^TA(2\lambda_rI-A)^{-1}e_1}\sigma^2
    \end{equation}
    where $A$ is the bidiagonal matrix defining the ODE system (\ref{eq:Infection:Pk}) from Proposition~\ref{prop:Infection_assumption} i.e., $A_{ii}~=~-i(N-i)\lambda_c$ and $A_{i+1,i}~=~i(N-i)\lambda_c$, and $w^T~=~\begin{bmatrix} 0 & \tfrac{1}{N-1} & \tfrac{2}{N-1} & \ldots & 1 \end{bmatrix}$, with $\lambda_r\neq0$.
\end{thm}
\begin{proof}
    Equation (\ref{eq:Infection:Pk}) defines a linear system $\dot P(s)~=~AP(s)$ with $P(s) = \begin{bmatrix} P_1(s) & P_2(s) & \ldots & P_N(s) \end{bmatrix}^T,$
    whose solution is given by
    $$P(s) = e^{As}P(0) \ \hbox{ with } \ P(0)=e_1.$$
    
    It follows from (\ref{eq:Infection:CDF}) that the \textit{cdf} and \textit{pdf} are given by
    $$F^{Inf}(s) = w^Te^{As}e_1 \ \hbox{ and } \ f^{Inf}(s) = w^TAe^{As}e_1.$$
    
    Injecting that result in the integral from Theorem~\ref{thm:general_bnd_pdf}, which we will see is well defined, one has
    $$\int_0^\infty f^{Inf}(s)\prt{1-e^{-2\lambda_rs}}\mathrm{d}s = 1-w^TA\int_0^\infty \prt{e^{As}e^{-2\lambda_rs}}\mathrm{d}se_1.$$
    
    Since $A$ and $-2\lambda_rI$ commute, the integral reduces to
    $$\int_0^\infty \prt{e^{As}e^{-2\lambda_rIs}}\mathrm{d}s = \int_0^\infty e^{(A-2\lambda_rI)s}\mathrm{d}s.$$
    
    Since $\lambda_r\neq0$, then $(A-2\lambda_rI)$ is invertible and has strictly negative eigenvalues, leading to
    $$\int_0^\infty e^{(A-2\lambda_rI)s}\mathrm{d}s = (A-2\lambda_rI)^{-1}\left[e^{(A-2\lambda_rI)s}\right]_0^\infty = (2\lambda_rI-A)^{-1}.$$
    
    Hence, the integral from Theorem~\ref{thm:general_bnd_pdf} is given by 
    $$\int_0^\infty f^{Inf}(s)\prt{1-e^{-2\lambda_rs}}\mathrm{d}s=1-w^TA(2\lambda_rI-A)^{-1}e_1,$$ 
    and the conclusion is direct from Theorem~\ref{thm:general_bnd_pdf}.
\end{proof}

\vspace{0.1cm}
\begin{cor}
    \label{cor:Inf:algebraic}
    The expression below is equivalent to (\ref{eq:Infection:Bound}).
    \begin{equation}
        \label{eq:Infection:Bound:Algebraic}
        \small
        \Ep{C(t)} \geq \frac{N-1}{N^2} \prt{ 1 - \frac{1}{1+\frac{N-1}{2} \frac{\lambda_c}{\lambda_r}} h\prt{N,\tfrac{\lambda_c}{\lambda_r}}} \sigma^2
    \end{equation}
    with 
    \begin{equation}
        \label{eq:Infection:h(N,L)}
        \small
        h(N,L) = \sum_{k=2}^N\prt{\frac{k-1}{N-1} \prod_{j=1}^{k-1}\frac{j(N-j)L}{2+(j+1)(N-j-1)L}}
    \end{equation}

\end{cor}
\vspace{0.1cm}
Corollary~\ref{cor:Inf:algebraic} presents an algebraic expression for the bound (\ref{eq:Infection:Bound}) from Theorem~\ref{thm:general_bnd_pdf}, and is proven in Appendix~\ref{ANNEX:proof:algebraic_Inf}. Detailed empirical explorations suggest that $h(N,L)~\leq~\tfrac{L}{2}~+~\tfrac{L}{2}\sum\nolimits_{k=1}^{N-2}\tfrac{2kL}{2+2kL}$. Using results on bounds in numerical integration, we can then find a compact bound on that expression, and derive the following empirical expression for bounding $\Ep{C(t)}$ from below.
\begin{equation}
    \label{eq:Infection:Bound:relaxed}
    \small%\scriptstyle
    \frac{N-1}{N^2} \prt{\frac{1}{1+\frac{N-1}{2}\frac{\lambda_c}{\lambda_r}}} \prt{1 + \tfrac{1}{2}\log\prt{\frac{1+(N-2)\tfrac{\lambda_c}{\lambda_r}}{1+\tfrac{\lambda_c}{\lambda_r}}}}\sigma^2
\end{equation}

%%%%%%%%%%%%%%%%%%%%%%%%%%%%%%%%%%%%%%%%%%%%%%%%%%%%%%%%%%%%%%%%%%%%%
%%%%%%%%%%%%%%%%%%%%%%%%%%%%%%%%%%%%%%%%%%%%%%%%%%%%%%%%%%%%%%%%%%%%%
\subsection{Discussion}
\label{Sec:Infection:discussion}

Fig.~\ref{fig:comp_all_3N_log} compares the bounds derived with the Ping and Infection models, and the expression (\ref{eq:Infection:Bound:relaxed}) based on the Infection bound, in logarithmic scale. The first observation is the improvement of both bounds from the Infection model, which relies on a weaker relaxation, in comparison with the Ping model. That improvement is even more apparent when the rate ratio $\lambda_c/\lambda_r$ increases, since the impact of the corresponding factor gets more important. In particular, the expression (\ref{eq:Infection:Bound:relaxed}) is actually exactly that of the bound from the Ping model multiplied by a logarithmic factor that scales to $\log N$ when $\lambda_c/\lambda_r$ is large. This could relate to the fact that information propagates in one hop in the Ping model, whereas it needs time with the Infection model.

\begin{figure}[thpb]
    \centering
    \includegraphics[width=0.45\textwidth]{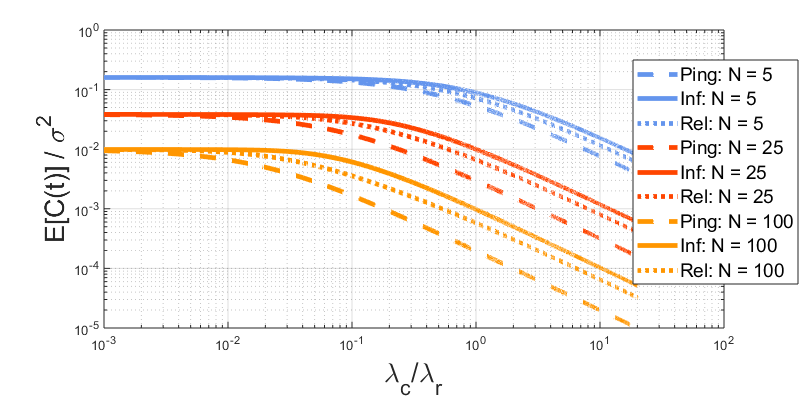}%{Pics/FINAL/COMP_LOG_2.png}
    \caption{\scriptsize Bounds derived with the Ping model (dashed), the Infection model (plain), and the expression (\ref{eq:Infection:Bound:relaxed}) (dotted) in terms of the rate ratio $\lambda_c/\lambda_r$, for different numbers of agents $N$, and with a unit variance $\sigma^2=1$.}
    \label{fig:comp_all_3N_log}
    \includegraphics[width=0.45\textwidth]{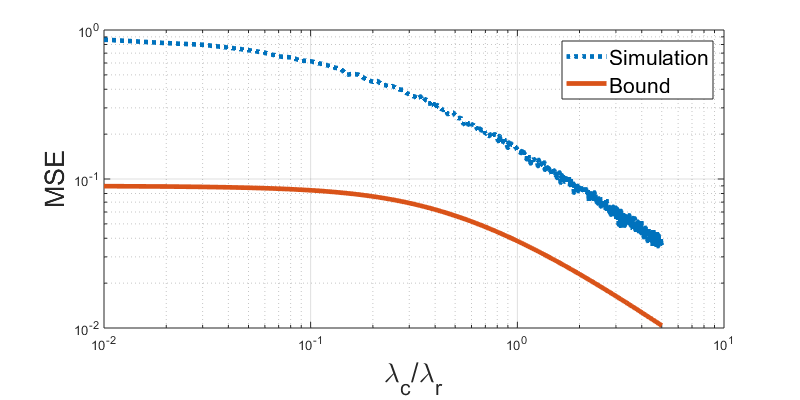}%{Pics/FINAL/SimVSBnd_Log_10000_v2.png}
    \caption{Performances comparison for solving an intrinsic consensus problem with $10$ agents between the bound from the Infection model (plain) and a simulation using a simple Gossip algorithm (dotted), both in logarithmic scale. The simulations are performed with $10 000$ iterations of the simulation over $200$ events, which is empirically shown to be enough for the system to be considered in steady state in that setting.}
    \label{fig:sim_log}
\end{figure}

Observe also in Fig.~\ref{fig:comp_all_3N_log} that the bound follows two distinct regimes: when $\lambda_c/\lambda_r$ is very small, it has a very low impact on the bound, which stays around $\tfrac{N-1}{N^2}\sigma^2$; whereas after some threshold, the bound decays according to $1/(\lambda_c/\lambda_r)$.

Finally, we have compared our bound with the results of the simplest version of average gossip \cite{Avg:Gossip}: agents take their own value $x_i$ as initial estimate $y_i(t)$, and whenever two agents $i,j$ communicate, their estimates become
\begin{equation}\label{eq:gossip}
    y_i(t^+) = y_j(t^+) = \tfrac{y_i(t^-)+y_j(t^-)}{2}.
\end{equation}
This algorithm computes the average in closed system. Observe that it is consistent with our setting presented in Section~\ref{Sec:Statement}, and that we did not make any adaptation to the open character. Fig.~\ref{fig:sim_log} compares the performances of the algorithm \eqref{eq:gossip} with the bound from the Infection model (\ref{eq:Infection:Bound}). It is interesting to notice that even though there is a gap between the curves, the behavior of the error in terms of $\lambda_c/\lambda_r$ is well captured by the bound, especially when that ratio gets large enough. Interestingly, the algorithm relies on only one variable, and does not make use of identifiers nor of any other information provided to the agents as described in Section~\ref{Sec:Statement}, while our bound is valid for algorithms potentially using this information. 
Hence one could wonder about the precise impact of that information and the efficiency gains it allows. 

%%%%%%%%%%%%%%%%%%%%%%%%%%%%%%%%%%%%%%%%%%%%%%%%%%%%%%%%%%%%%%%%%%%%%%%%%%%%%%%%
%%%%%%%%%%%%%%%%%%%%%%%%%%%%%%%%%%%%%%%%%%%%%%%%%%%%%%%%%%%%%%%%%%%%%%%%%%%%%%%%
%%%%%%%%%%%%%%%%%%%%%%%%%%%%%%%%%%%%%%%%%%%%%%%%%%%%%%%%%%%%%%%%%%%%%%%%%%%%%%%%
\section{Conclusions}
\label{Sec:Ccl}

We considered in this paper the possibility of arrivals and departures of agents in the study of multi-agent systems, and highlighted several challenges arising from this property. In particular, it prevents algorithms to converge for many problems, making their analysis challenging.

We focused on the analysis of the performances of intrinsic averaging algorithms for open systems of constant size, for which we derived \textit{fundamental performance limitations} as lower bounds on the Mean Square Error of estimation in steady state. This was done by studying the performances of an algorithm optimal in a context more favorable to the agents than what is usually assumed in multi-agent systems, and relied then on the relaxation of the information spreading mechanism in the system. Two of such relaxations were studied, leading to a conservative but readable bound, and to a tighter but more complex bound.

Possible extensions include the study of more complicated problems (e.g. decentralized optimization \cite{DO:dualAvg}), more structured constraints on the inter-agent communications, and the consideration of variable-size systems. More fundamentally, our bounds purely rely on limitations on the information propagation, and are valid even if the agents have some strong knowledge on the system and use identifiers or heterogeneous algorithms. Stronger lower bounds can be obtained under more restrictive assumptions, e.g. by improving the Infection model from Section~\ref{Sec:Infection} with healing. Yet, it would be interesting to investigate what factors significantly impact the bound: in particular, an anonymous algorithm was already shown in Section~\ref{Sec:Infection:discussion} to exhibit performances not too different from our bound, questioning the impact of anonymity of the agents.

%%%%%%%%%%%%%%%%%%%%%%%%%%%%%%%%%%%%%%%%%%%%%%%%%%%%%%%%%%%%%%%%%%%%%%%%%%%%%%%%
%%%%%%%%%%%%%%%%%%%%%%%%%%%%%%%%%%%%%%%%%%%%%%%%%%%%%%%%%%%%%%%%%%%%%%%%%%%%%%%%
%%%%%%%%%%%%%%%%%%%%%%%%%%%%%%%%%%%%%%%%%%%%%%%%%%%%%%%%%%%%%%%%%%%%%%%%%%%%%%%%
\bibliographystyle{ieeetr}
\bibliography{ref}

%%%%%%%%%%%%%%%%%%%%%%%%%%%%%%%%%%%%%%%%%%%%%%%%%%%%%%%%%%%%%%%%%%%%%%%%%%%%%%%%
%%%%%%%%%%%%%%%%%%%%%%%%%%%%%%%%%%%%%%%%%%%%%%%%%%%%%%%%%%%%%%%%%%%%%%%%%%%%%%%%
%%%%%%%%%%%%%%%%%%%%%%%%%%%%%%%%%%%%%%%%%%%%%%%%%%%%%%%%%%%%%%%%%%%%%%%%%%%%%%%%
%\newpage
\appendix

%%%%%%%%%%%%%%%%%%%%%%%%%%%%%%%%%%%%%%%%%%%%%%%%%%%%%%%%
\subsection{Proof of Corollary~\ref{cor:Inf:algebraic}}
\label{ANNEX:proof:algebraic_Inf}

We prove in this Section that the expression (\ref{eq:Infection:Bound}) is equivalent to the following algebraic expression
\begin{equation*}
   \Ep{C(t)} \geq \frac{N-1}{N^2} \prt{ 1 - \frac{1}{1+\frac{N-1}{2} \frac{\lambda_c}{\lambda_r}} h\prt{N,\tfrac{\lambda_c}{\lambda_r}}} \sigma^2
\end{equation*}
with 
\begin{equation*}
    h(N,L) = \sum_{k=2}^N\prt{\frac{k-1}{N-1} \prod_{j=1}^{k-1}\frac{j(N-j)L}{2+(j+1)(N-j-1)L}}.
\end{equation*}
\vspace{0.5cm}

We will use the following standard result on the inverse of bidiagonal matrices.

\begin{lem}
    \label{lem:inv_bidiag_unit}
    The inverse $A^{-1}$ of a bidiagonal matrix with a unit diagonal 
    \begin{equation*}
        A = 
        \begin{bmatrix} 
        1 & 0 & \ldots & 0 & 0\\
        -a_1 & 1 &  & 0 & 0\\
        \vdots &\ddots & \ddots & & \vdots\\
        0 &  & \ddots & 1 & 0\\
        0 & 0 & \ldots & -a_{N-1} & 1
        \end{bmatrix}
    \end{equation*}
is  lower triangular matrix with unit diagonal and
    $$A^{-1}_{i,j} = \prod_{k=j}^{i-1} a_k \ \ (\forall i>j)$$
\end{lem}

\vspace{0.1cm}

We also need the following technical lemma.
\begin{lem}
    \label{lem:property}
    For all matrix $A$, and $\forall \beta \in \mathbb{R}$ such that $(\beta I~-~A)^{-1}$ is well defined, there holds
    \begin{equation}
        A(\beta I-A)^{-1} = (\beta I-A)^{-1}\beta-I
    \end{equation}
\end{lem}
\begin{proof}
    Let $\tilde A := A-\beta I$, then 
    $$A(\beta I-A)^{-1} = -(\tilde A+\beta I)\tilde A^{-1} = -I-\beta \tilde A^{-1}.$$
    
    The conclusion follows that
    $$-\tilde A^{-1} = (\beta I-A)^{-1}.$$
\end{proof}

\vspace{0.2cm}

We start from the result of Theorem~\ref{thm:Infection:Bound}:
$$\Ep{C(t)} \geq \tfrac{N-1}{N^2} \prt{1-w^TA(2\lambda_rI-A)^{-1}e_1}\sigma^2.$$
Since $(2\lambda_rI-A)$ is bidiagonal and invertible when $\lambda_r\neq0$, one writes $(2\lambda_rI-A) = PQ$ with 
\begin{itemize}
    \item $P$ a diagonal matrix with $P_{ii} = i(N-i)\lambda_c+2\lambda_r$;
    \item $Q$ a bidiagonal matrix of unit diagonal with $Q_{i+1,i}~=~\tfrac{-i(N-i)\lambda_c}{2\lambda_r+(i+1)(N-i-1)\lambda_c}$.
\end{itemize}

From this definition, one has
$$(2\lambda_r-A)^{-1} = Q^{-1}P^{-1}$$ where $P^{-1}_{ii} = \tfrac{1}{2\lambda_r+i(N-i)\lambda_c}$, and where $Q^{-1}$ is obtained using Lemma~\ref{lem:inv_bidiag_unit} with $$a_k~=~\tfrac{k(N-k)\lambda_c}{2\lambda_r+(k+1)(N-k-1)\lambda_c}.$$

From Lemma~\ref{lem:property} and the previous result, there holds
$$w^TA(2\lambda_rI-A)^{-1}e_1 = w^T[Q^{-1}P^{-1}2\lambda_r-I]e_1$$
with
$$w^TIe_1 = w^Te_1 = 0.$$

Then, since $P^{-1}e_1 = P^{-1}_{11}e_1$ and $P^{-1}_{11} = \tfrac{1}{2\lambda_r+(N-1)\lambda_c}$,
$$w^TQ^{-1}P^{-1}2\lambda_re_1 = \tfrac{2\lambda_r}{2\lambda_r+(N-1)\lambda_c}w^TQ^{-1}e_1.$$

From the definition of $Q^{-1}$ and thus of $a_k$, one has then
$$w^TA(2\lambda_rI-A)^{-1}e_1=\tfrac{2\lambda_r}{2\lambda_r+(N-1)\lambda_c}\sum\nolimits_{k=1}^N\prt{\tfrac{k-1}{N-1}\prod\nolimits_{j=0}^{k-1}a_j}.$$

Algebraic manipulations of this result within equation (\ref{eq:Infection:Bound}) from Theorem~\ref{thm:Infection:Bound} allow to conclude.

\end{document}